\documentclass[journal,twocolumn]{IEEEtran}
\usepackage{bbm}
\usepackage{array}
\usepackage{dcolumn}
\usepackage{epsfig}
\usepackage[intlimits]{amsmath}
\usepackage{amsmath, amsfonts, yhmath, bm}
\usepackage{amssymb}
\usepackage{color,soul}
\usepackage[dvipsnames]{xcolor}
\usepackage[normalem]{ulem}
\usepackage{enumerate}
\usepackage{stackengine}
\usepackage[noadjust]{cite}
\usepackage{graphicx}
\usepackage{caption}
\usepackage{subcaption}
\usepackage[font=footnotesize]{subcaption}
\usepackage{multirow}
\usepackage{cite}
\usepackage[font=footnotesize]{caption}
\usepackage{etoolbox}
\usepackage{tcolorbox}
\usepackage{float}
\usepackage{dsfont}
\usepackage{tikz}
\usetikzlibrary{arrows}
\usepackage{comment}

\usepackage{hyperref}
\hypersetup{
    colorlinks,
    linkcolor={red!50!black},
    citecolor={blue!50!black},
    urlcolor={blue!80!black}
}

\usepackage{balance}
\usepackage{url}
\usepackage[inline]{enumitem} % For inline enumerations

\usepackage{vmr-symbols-vecbold}
\usepackage{standard-macros}
\PassOptionsToPackage{hyphens}{url}
\usepackage{hyperref}

\newcommand{\defeq}{\mathrel{\mkern-0.25mu=}}
\DeclareSymbolFontAlphabet{\amsmathbb}{AMSb}%

\newcommand{\lro}[1]{\lefto({#1}\right)}																% left right paranthesis operator
\newcommand{\lr}[1]{\left({#1}\right)}																% left right paranthesis operator
\newcommand{\lrb}[1]{\left \lbrace {#1} \right \rbrace}															% left right braces operator
\safemath{\dopplerspread}{B_D}																								% doppler spread
\safemath{\delayspread}{T_D}																									% delay spread
\safemath{\nc}{n\sub{c}}																										% coherence time
\safemath{\nf}{n\sub{f}}																										% feedback message length
\safemath{\efa}{p\sub{sc}}
\safemath{\efb}{p\sub{cs}}
\safemath{\ef}{\epsilon\sub{f}	}
\safemath{\nd}{n\sub{d}}																										% data symbols
\safemath{\ntx}{n\sub{t}} 																											% transmit antennas
\safemath{\nrx}{n\sub{r}}																											% receive antennas
\safemath{\ntxt}{\tilde{n\sub{t}}}																											% receive antennas
\safemath{\cb}{\ensuremath{L}} 																								% code blocks
\safemath{\cl}{\ensuremath{n}} 																								% codelength
\safemath{\txanto}{{\ensuremath{\tilde{m}_t}}} 																		% transmit antennas when some is turned off
\safemath{\cs}{M} 																														% code size
\safemath{\idPustm}{\ensuremath{S_{k}}}
\safemath{\error}{\ensuremath{\epsilon}} 																				%Error target
\safemath{\eexp}{\ensuremath{\mathcal{E}}} 																			%Error exponent
\safemath{\nsubc}{n\sub{s}}			 																						% number of subcarriers
\safemath{\nofdm}{n\sub{o}} 																									% number of OFDM symbols
\safemath{\bc}{\ensuremath{B_c}} 																							% Coherence bandwidth
\safemath{\ts}{\ensuremath{T_s}} 																							% Symbol time
\safemath{\nrb}{\ensuremath{n_{rb}}} 																						% Symbol time
\safemath{\rul}{\ensuremath{\rho\sub{ul}}}
\safemath{\rdl}{\ensuremath{\rho\sub{dl}}}

\safemath{\nres}{\ell}
\safemath{\nr}{n\sub{r}}
\safemath{\maxk}{M^*\lr{\nres, \nsubc, \nofdm, \epsilon, \rho}}
\safemath{\Rmax}{R^*}%\lr{\nres, \nsubc, \nofdm,M, \epsilon, \rho}}
\safemath{\Emin}{E\sub{b}^*/N_0}%\lr{\nres, \nsubc, \nofdm,M, \epsilon, \rho}}
\safemath{\Eminf}{\frac{E\sub{b}^*}{N_0}}
\safemath{\np}{\ensuremath{n\sub{p}}}
\safemath{\ndf}{\ensuremath{\bar{n}\sub{d}}}
\safemath{\npf}{\ensuremath{\bar{n}\sub{p}}}
\safemath{\code}{\ensuremath{\mathcal{C}}}
\safemath{\err}{\ensuremath{\epsilon}}
\safemath{\rp}{\ensuremath{\rho\sub{p}}}
\safemath{\rd}{\ensuremath{\rho\sub{d}}}
\safemath{\cohtime}{\ensuremath{T\sub{c}}}
\safemath{\cohbw}{\ensuremath{B\sub{c}}}
\safemath{\nmax}{\ensuremath{\ell\sub{m}}}
\safemath{\ntot}{\ensuremath{n\sub{tot}}}
\safemath{\nul}{\ensuremath{n\sub{ul}}}
\safemath{\ndl}{\ensuremath{n\sub{dl}}}

\safemath{\yp}{\ensuremath{\randvecy_{\nu}^{(\text{p})}}}
\safemath{\yd}{\ensuremath{\randvecy_{\nu}^{(\text{d})}}}
\safemath{\ypd}{\ensuremath{\vecy_{\nu}^{(\text{p})}}}
\safemath{\ydd}{\ensuremath{\vecy_{\nu}^{(\text{d})}}}

\safemath{\ypf}{\ensuremath{\bar{\randvecy}_{\nu}^{(\text{p})}}}
\safemath{\ydf}{\ensuremath{\bar{\randvecy}_{\nu}^{(\text{d})}}}
\safemath{\ypdf}{\ensuremath{\bar{\vecy}_{\nu}^{(\text{p})}}}
\safemath{\yddf}{\ensuremath{\bar{\vecy}_{\nu}^{(\text{d})}}}

\safemath{\xp}{\ensuremath{\vecx^{(\text{p})}}}
\safemath{\xd}{\ensuremath{\randvecx_{\nu}^{(\text{d})}}}
\safemath{\xdd}{\ensuremath{\vecx_{\nu}^{(\text{d})}}}

\safemath{\xpf}{\ensuremath{\bar{\vecx}^{(\text{p})}}}
\safemath{\xdf}{\ensuremath{\bar{\randvecx}_{\nu}^{(\text{d})}}}
\safemath{\xddf}{\ensuremath{\bar{\vecx}_{\nu}^{(\text{d})}}}

\safemath{\xdb}{\ensuremath{\overline{\randvecx}^{(\text{d})}}}
\safemath{\Pxd}{\ensuremath{P_{\randvecx^{(\text{d})}}}}

\safemath{\xpbar}{\ensuremath{\overline{\matX}^{(\text{p})}}}
\safemath{\xdbar}{\ensuremath{\overline{\randmatX}^{(\text{d})}}}

\safemath{\xdv}{\ensuremath{\randvecx^{(\text{d})}}}
\safemath{\xdbarv}{\ensuremath{\overline{\randvecx}^{(\text{d})}}}
\safemath{\ydv}{\ensuremath{\randvecy^{(\text{d})}}}

\safemath{\xdr}{\ensuremath{\matX^{(\text{d})}}}

\safemath{\ttx}{\ensuremath{\tau\sub{tx}}}
\safemath{\trx}{\ensuremath{\tau\sub{rx}}}
\safemath{\ack}{\ensuremath{\mathrm{s}}}
\safemath{\nack}{\ensuremath{\mathrm{c}}}

\newcommand{\prob}[1]{\ensuremath{\mathbb{P}\lro{#1}}}

\safemath{\mI}{\ensuremath{i\lro{\randvecy ; \randvecx}}} 				% i(Y;X)
\safemath{\randveca}{\bm{A}}
\safemath{\randvecb}{\bm{B}}
\safemath{\randvecc}{\bm{C}}
\safemath{\randvecd}{\bm{D}}
\safemath{\randvece}{\bm{E}}
\safemath{\randvecf}{\bm{F}}
\safemath{\randvecg}{\bm{G}}
\safemath{\randvech}{\bm{H}}
\safemath{\randveci}{\bm{I}}
\safemath{\randvecj}{\bm{J}}
\safemath{\randveck}{\bm{K}}
\safemath{\randvecl}{\bm{L}}
\safemath{\randvecm}{\bm{M}}
\safemath{\randvecn}{\bm{N}}
\safemath{\randveco}{\bm{O}}
\safemath{\randvecp}{\bm{P}}
\safemath{\randvecq}{\bm{Q}}
\safemath{\randvecr}{\bm{R}}
\safemath{\randvecs}{\bm{S}}
\safemath{\randvect}{\bm{T}}
\safemath{\randvecu}{\bm{U}}
\safemath{\randvecv}{\bm{V}}
\safemath{\randvecw}{\bm{W}}
\safemath{\randvecx}{\bm{X}}
\safemath{\randvecy}{\bm{Y}}
\safemath{\randvecz}{\bm{Z}}
\safemath{\randvecphi}{\bm{\Phi}}

\safemath{\randmatA}{\amsmathbb{A}}
\safemath{\randmatB}{\amsmathbb{B}}
\safemath{\randmatC}{\amsmathbb{C}}
\safemath{\randmatD}{\amsmathbb{D}}
\safemath{\randmatE}{\amsmathbb{E}}
\safemath{\randmatF}{\amsmathbb{F}}
\safemath{\randmatG}{\amsmathbb{G}}
\safemath{\randmatH}{\amsmathbb{H}}
\safemath{\randmatI}{\amsmathbb{I}}
\safemath{\randmatJ}{\amsmathbb{J}}
\safemath{\randmatK}{\amsmathbb{K}}
\safemath{\randmatL}{\amsmathbb{L}}
\safemath{\randmatM}{\amsmathbb{M}}
\safemath{\randmatN}{\amsmathbb{N}}
\safemath{\randmatO}{\amsmathbb{O}}
\safemath{\randmatP}{\amsmathbb{P}}
\safemath{\randmatQ}{\amsmathbb{Q}}
\safemath{\randmatR}{\amsmathbb{R}}
\safemath{\randmatS}{\amsmathbb{S}}
\safemath{\randmatT}{\amsmathbb{T}}
\safemath{\randmatU}{\amsmathbb{U}}
\safemath{\randmatV}{\amsmathbb{V}}
\safemath{\randmatW}{\amsmathbb{W}}
\safemath{\randmatX}{\amsmathbb{X}}
\safemath{\randmatY}{\amsmathbb{Y}}
\safemath{\randmatZ}{\amsmathbb{Z}}
\safemath{\randmatSigma}{\mathbb{\Sigma}}
\safemath{\randmatPhi}{\mathbb{\Phi}}
\safemath{\randmatLambda}{\mathbb{\Lambda}}

\safemath{\matSigma}{\bm{\Sigma}}
\safemath{\matPhi}{\bm{\Phi}}
\safemath{\matLambda}{\bm{\Lambda}}

\newcommand{\midsub}{\text{\tiny MID}}
\newcommand{\mie}{\text{\tiny MIE}}
\newcommand{\fa}{\text{\tiny FA}}
\newcommand{\md}{\text{\tiny MD}}

\newcommand{\GLRT}{\text{\tiny GLRT}}
\newcommand{\mmse}{\text{\tiny MMSE}}
\newcommand{\eff}{\text{\tiny eff}}

\makeatletter
\newsavebox\myboxA
\newsavebox\myboxB
\newlength\mylenA

\newcommand*\mybar[2][0.75]{%
	\sbox{\myboxA}{$\m@th#2$}%
	\setbox\myboxB\null% Phantom box
	\ht\myboxB=\ht\myboxA%
	\dp\myboxB=\dp\myboxA%
	\wd\myboxB=#1\wd\myboxA% Scale phantom
	\sbox\myboxB{$\m@th\overline{\copy\myboxB}$}%  Overlined phantom
	\setlength\mylenA{\the\wd\myboxA}%   calc width diff
	\addtolength\mylenA{-\the\wd\myboxB}%
	\ifdim\wd\myboxB<\wd\myboxA%
	\rlap{\hskip 0.5\mylenA\usebox\myboxB}{\usebox\myboxA}%
	\else
	\hskip -0.5\mylenA\rlap{\usebox\myboxA}{\hskip 0.5\mylenA\usebox\myboxB}%
	\fi}
\makeatother
\makeatletter
\renewcommand*{\defeq}{\mathrel{\rlap{%
			\raisebox{0.3ex}{$\m@th\cdot$}}%
		\raisebox{-0.3ex}{$\m@th\cdot$}}%
	=}
\makeatother

\newtheorem{proposition}{Proposition}
\begin{document}

\title{Extremum-Based Joint Compression and Detection for Distributed Sensing}

\author{Amir~Weiss$^{\star}$ and Alejandro Lancho$^{\dagger}$\\
$^{\star}$Faculty of Engineering, Bar-Ilan University, Ramat Gan, Israel\\
Signal Theory and Communications Department, Universidad Carlos III de Madrid, Leganes, Spain\\
amir.weiss@biu.ac.il, alancho@ing.uc3m.es
% \thanks{A. Weiss is with the Faculty of Engineering, Bar-Ilan University, Ramat Gan, 5290002, Israel, e-mail: amir.weiss@biu.ac.il. A. Lancho is with the Signal Theory and Communications Department, Universidad Carlos III de Madrid, Leganes, 28911, Spain, and with the Gregorio Marañón Health Research Institute, Madrid, 28007, Spain; e-mail: alancho@ing.uc3m.es.}% <-this % stops a space
% \thanks{J. Doe and J. Doe are with Anonymous University.}% <-this % stops a space
\thanks{A. Lancho is also with the Gregorio Marañón Health Research Institute, Madrid, 28007, Spain. He received funding from the Comunidad de Madrid under Grant Agreements No.~2023-T1/COM-29065 (César Nombela program) and TEC-2024/COM-89, and from the Ministerio de Ciencia, Innovación y Universidades, Spain under Grant Agreement No.~PID2023-148856OA-I00.}
}
% \thanks{Manuscript received April 19, 2005; revised August 26, 2015.}}


\markboth{}{}



\maketitle

\begin{abstract}
We study joint compression and detection in distributed sensing systems motivated by emerging applications such as IoT-based localization. Two spatially separated sensors observe noisy signals and can exchange only a $k$-bit message over a reliable one-way low-rate link. One sensor compresses its observation into a $k$-bit description to help the other decide whether their observations share a common underlying signal or are statistically independent. We propose a simple extremum-based strategy, in which the encoder sends the index of its largest sample and the decoder performs a scalar threshold test. We derive exact nonasymptotic false-alarm and misdetection probabilities and validate the analysis with representative simulations.
\end{abstract}
\vspace{-0.2cm}


\IEEEpeerreviewmaketitle

\vspace{-0.7cm}
\section{Introduction}
% \vspace{-0.2cm}
Modern communication and sensing systems operate under growing pressure. The electromagnetic spectrum is increasingly crowded, and services compete for bandwidth under stringent demands on latency, reliability, and energy efficiency. Furthermore, when it comes to distributed sensing applications, such as large-scale IoT deployments, an additional difficulty arises: each sensor acquires only \emph{local, noisy, and incomplete measurements} of the environment. When resources are scarce, to meet reliability demands, such partial observations cannot simply be compensated for by sending more data.% or repeating transmissions.

Communication systems have traditionally been designed so that the receiver can \emph{reconstruct the entire transmitted signal}, even when only a small portion of that information is relevant for a downstream decision. An illustrative example is image transmission, where compression and channel coding are optimized to preserve every sample or pixel with high fidelity, even if the final task requires only a (possibly tiny) subset of that information. This kind of mismatch becomes particularly limiting in distributed sensing. Exchanging all locally observed data across sensors would be equivalent to transmitting every pixel of an image just to identify a single object. Under bandwidth, latency, and energy constraints, such strategies are impractical. Hence, classical centralized detectors, which rely on full access to all observations to implement procedures such as the generalized likelihood ratio test (GLRT)~\cite{kay2009fundamentals}, become infeasible in these constrained settings.

Motivated by these challenges, we study a fundamental instance of \emph{task-oriented communication} in distributed sensing, namely joint compression and detection~\cite{guleryuz2005joint,chan2005joint,viswanathan2002distributed,chamberland2007wireless}. In particular, we consider the case where two spatially separated sensors must decide whether their observations correspond to noisy, delayed versions of the same underlying signal or whether they are independent, while being restricted to a unidirectional communication budget of only $k$ bits. This abstraction captures a class of collaborative detection problems in which raw data exchange is infeasible but accurate decisions about statistical dependence or source association are still required.%and enables extremum-based strategies that mirror centralized detectors while remaining compatible with modern sensing architectures~\cite{cui2007estimation}.

Hence, the central question guiding this work is: \emph{How can a pair of spatially separated sensors cooperate to perform reliable detection under tight communication constraints?} Driven by this motivation, we develop a joint compression--detection scheme and derive \emph{exact} closed-form expressions for the false-alarm (FA) and misdetection (MD) probabilities, which constitute a detector's full statistical characterization.

% \subsection{Background and Motivation}
% Communication constraints between spatially separated sensors arise in many modern applications, including device-to-device connectivity in internet of things (IoT) networks~\cite{sarkar2014diat}, distributed radar systems~\cite{liang2011design}, and environmental monitoring. As the scale of sensing systems grows, there is a pressing need for accurate inference based on highly compressed information exchanges.

% A key challenge in this setting is collaborative detection: sensors with noisy, partial observations must decide on the presence of a target signal. Classical centralized approaches, which rely on full access to all data for optimal rules such as the likelihood ratio test (LRT)~\cite{van2004detection}, become impractical when communication is limited. This motivates developing joint compression and detection strategies that operate under strict information constraints (e.g.,~\cite{guleryuz2005joint}).

% In this work, we focus on a fundamental instance of this problem: two sensors must detect whether their observations are noisy, delayed versions of a common signal, or independent, while limited to unidirectional $k$-bit communication. We propose an extremum-based scheme that mirrors the structure of the centralized solution, making it suitable for modern sensing applications~\cite{cui2007estimation}.

\vspace{-0.25cm}
\section{Problem Formulation}\label{sec:problemformulation}
Consider two spatially separated sensors, each of which observes a signal. Sensor~1, the encoder, communicates with Sensor~2, the decoder, over a one-way noiseless link, over which it can convey a $k$-bit message. Sensor~2 is co-located with the computing unit and must decide, based on the received message and its local observation, whether the observed signals are statistically independent or noisy, shifted versions of a common source (a prerequisite for, e.g., \cite{weiss2025joint}).

Let $\{\rndx[n]\}_{n \in \integers}$ and $\{\rndy[n]\}_{n \in \integers}$ denote the observations at the encoder (Sensor~1) and the decoder (Sensor~2), respectively. We model the signals via an underlying source process $\{\rnds[n]\}_{n\in\integers}$ and noise processes $\{\rndz_1[n]\}_{n\in\integers}, \{\rndz_2[n]\}_{n\in\integers}$, where
\begin{equation}
\rnds[n] \overset{\text{iid}}{\sim} \normal(0,\sigma_{\rnds}^2), \quad
\rndz_i[n] \overset{\text{iid}}{\sim} \normal(0,\sigma_i^2), \; i \in \{1,2\},
\end{equation}
and $\{\rnds[n]\}$, $\{\rndz_1[n]\}$, and $\{\rndz_2[n]\}$ are mutually statistically independent, where $\text{iid}$ stands for independent, identically distributed. The detection task consists in deciding between the following two hypotheses:
\begin{align}
    &\setH_0: \begin{cases}
        \rndx[n] = \rndz_1[n] \\
        \rndy[n] = \rndz_2[n]
    \end{cases}\hspace{-0.3cm},\;\quad\quad\quad\quad\;\;\text{(the null hypothesis)}\label{eq:H0_model}\\[0.1cm]
    &\setH_1: \begin{cases}
        \rndx[n] = \rnds[n] + \rndz_1[n] \\
        \rndy[n] = \rnds[n-d] + \rndz_2[n]
    \end{cases}\hspace{-0.3cm},\;\text{(the alternative hypothesis)}\label{eq:H1_model}
\end{align}
where $d \in \setD \triangleq \{-d_m,\ldots,d_m\}$ is an unknown deterministic time delay, $d_m\in\positivereals$ being the (absolute) maximal delay, and $\sigma_1,\sigma_2,\sigma_{\rnds} \in \reals_{+}$ are deterministic unknowns.

By construction, it follows that the conditional cross-correlation between the two observations satisfies
\begin{align}
 \Exop\left[ \rndy[n]\rndx[n-\ell] \mid \setH \right]
 &= \begin{cases}
     0, & \hspace{-0.2cm}\setH = \setH_0, \\[0.05cm]
     R_{\rnds\rnds}[\ell-d], & \hspace{-0.2cm} \setH = \setH_1,
 \end{cases}
\end{align}
where $R_{\rnds\rnds}[\ell] \triangleq \Exop[ \rnds[n]\rnds[n-\ell] ]
    = \sigma_{\rnds}^2 \cdot \mathbbm{1}_{\{\ell=0\}}$ denotes the autocorrelation of the white process $\{\rnds[n]\}_{n\in\integers}$.

The encoder observes $\setX_N \triangleq \{\rndx[n]\}_{n=0}^{N-1}$ and transmits
\begin{equation}
    \rvecm = f(\setX_N) \in \{0,1\}^{k\times 1},
\end{equation}
 a $k$-bit message, where $f: \reals^{N\times 1} \to \{0,1\}^{k\times 1}$ denotes the (block) encoding rule. The decoder observes $\{\rndy[n]\}_{n\in\integers}$ and the received message $\rvecm$, and applies a decision rule
\begin{equation}
    \delta: \{0,1\}^{k\times 1} \times \reals^{N\times 1} \to \{\setH_0,\setH_1\},
\end{equation}
so that the final decision is $\delta(\rvecm,\setY_{N})$, where $\setY_{N}$ is a batch of $N$ samples of $\{\rndy[n]\}$, which can be chosen by the decoder. We assume equal priors $\prob{\setH=\setH_0} = \prob{\setH=\setH_1} = \tfrac{1}{2}$ and aim to minimize the overall error probability $\prob{ \delta(\rvecm, \setY_{N}) \neq \setH }$.

More specifically, for a prescribed tolerance on the FA probability—declaring $\setH_1$ when $\setH_0$ holds—we seek to minimize the MD probability—declaring $\setH_0$ when $\setH_1$ holds. Formally, for a given FA level $\alpha \in [0,1]$, we consider
\begin{align}
    \hspace{-0.2cm}&\inf_{ \substack{ f: \reals^{N\times 1} \to \{0,1\}^{k\times 1} \\ \delta: \{0,1\}^{k\times 1} \times \reals^{N\times 1} \to \{\mathcal{H}_0, \mathcal{H}_1\} } } &&\prob{ \delta = \setH_0 \mid \setH_1 } \\
    \hspace{-0.2cm}&\quad\quad\quad\quad\;\;\;\mathrm{s.t.} &&\prob{ \delta = \setH_1 \mid \setH_0 }\leq \alpha.
\end{align}
This Neyman--Pearson formulation serves as a benchmark for designing low-complexity, $k$-bit-constrained schemes.

\vspace{-0.2cm}
\section{Main Results}\label{sec:mainresults}
We now present our main results. First, we describe the proposed joint compression--detection scheme operating under the $k$-bit communication constraint, and then develop a nonasymptotic performance analysis that characterizes its behavior at finite blocklength. In particular, we derive exact expressions for the resulting FA and MD probabilities, thus providing a full and accurate characterization of our scheme.

\vspace{-0.25cm}
\subsection{Proposed Joint Compression and Detection Scheme}\label{eq:subsecproposedscheme}

Below is our proposed, particularly simple, $k$-bit scheme, referred to as the \emph{maximum-index-based detector} (MID). At a high level, the encoder conveys only the location of its largest sample, and the decoder then inspects a small neighborhood around this location in its own observation and compares the largest aligned sample to a threshold.

\paragraph*{Scheme definition}
Fix a blocklength $N = 2^k$, an admissible delay set $\setD = \{-d_m,\ldots,d_m\}$, and a threshold $\tau \in \positivereals$. The scheme operates as follows.

\begin{enumerate}
    \item \textbf{Encoder (Sensor~1):}
    The encoder observes the block $\setX_N$ and computes the index of its maximum,
    \begin{equation}
        \rndj \triangleq \arg\max_{0 \leq n \leq N-1} \rndx[n].
    \end{equation}
    Under the continuous Gaussian model, this maximum is almost surely unique; otherwise, any fixed tie-breaking rule is used. The encoder then maps $\rndj$ to its $k$-bit binary representation, and transmits $\rvecm$ to the decoder.

    \item \textbf{Decoder (Sensor~2):}
    The decoder receives $\rvecm$ and recovers the index $\rndj$. Since it observes $\{\rndy[n]\}_{n\in\integers}$, it can in particular access the samples
    \begin{equation}
        \setY_{N}^{\setD} \triangleq
        \{\rndy[n]\}_{n=-d_m}^{N-1+d_m}.
    \end{equation}
    The decoder forms the test statistic
    \begin{equation}\label{eq:testdef-mid}
        \rndt_{\midsub} \triangleq \max_{\ell \in \setD} \rndy[\rndj + \ell],
    \end{equation}
    i.e., the largest sample of the process $\rndy[n]$ in the admissible delay set
    around the reported extremum index $\rndj$.

    \item \textbf{Decision rule:}
    The final decision is obtained by a scalar threshold test:
    \begin{equation}\label{eq:miedecisionrule}
        \delta_{\midsub}(\rvecm,\setY_{N}^{\setD}) \triangleq
        \begin{cases}
            \setH_1, & \text{if } \rndt_{\midsub} \geq \tau, \\[0.05cm]
            \setH_0, & \text{if } \rndt_{\midsub} < \tau,
        \end{cases}
    \end{equation}
    or, equivalently,
    \begin{equation}\label{eq:T_MIE_repeat}
        \max_{\ell \in \setD} \rndy[\rndj + \ell]
        \;\mathop{\gtrless}_{\setH_0}^{\setH_1}\;
        \tau,
    \end{equation}
    where the threshold $\tau$ is chosen to meet a prescribed FA level, as characterized in the analysis in Section~\ref{subsec:nonasymptoticanalysis}.
\end{enumerate}

In the next subsection we develop an intuitive interpretation of this scheme by relating it to the GLRT without communication constraints: we show that the latter reduces to a threshold test on an empirical cross-correlation statistic, and that our extremum-based rule can similarly be viewed as a scaled empirical cross-correlation computed under a $k$-bit constraint. We then proceed to derive an exact nonasymptotic characterization of the resulting FA and MD probabilities.

\vspace{-0.2cm}
\subsection{The GLRT in a Constraint-Free Communication Setting and Its Relation to the Proposed Scheme}

Consider an \emph{ideal} setting, without any communication constraint, in which a fusion center has full access to both observation blocks $\setX_N$ and $\setY_N^{\setD}$. In this case, a natural benchmark is the GLRT for deciding between independence and a time-shifted common source. For a fixed time-delay $d \in \setD$, define the aligned pairs $\left\{(\rndx[n],\rndy[n+d])\right\}_{n=0}^{N-1}$.
% \begin{equation}
%     \left\{(\rndx[n],\rndy[n+d])\right\}_{n=0}^{N-1}.
% \end{equation}

From the viewpoint of the fusion center, only the joint second-order statistics (SOSs) of these pairs matter. Under $\setH_0$, they are independent zero-mean Gaussian with some (\textit{a priori} unknown) marginal variances and zero correlation; under $\setH_1$ with delay $d$, they are still zero-mean jointly Gaussian, but with a nonzero correlation. Thus, for each $d$ we can write
\begin{equation}
    \tp{[\rndx[n] \; \rndy[n+d]]} \sim \normal\big(0, \Sigma_i(d)\big),
\end{equation}
where, for $i\in\{0,1\}$,
\begin{equation}
\Sigma_0(d) =
\begin{bmatrix}
\sigma_{\rndx,0}^2 & 0 \\
0                  & \sigma_{\rndy,0}^2
\end{bmatrix},
\quad
\Sigma_1(d) =
\begin{bmatrix}
\sigma_{\rndx,1}^2 & \rho_d \sigma_{\rndx,1}\sigma_{\rndy,1} \\
\rho_d \sigma_{\rndx,1}\sigma_{\rndy,1} & \sigma_{\rndy,1}^2
\end{bmatrix},
\end{equation}
with $\sigma_{\rndx,0}^2,\sigma_{\rndy,0}^2,\sigma_{\rndx,1}^2,\sigma_{\rndy,1}^2>0$
and $\rho_d\in(-1,1)$ treated as unknown nuisance parameters. The only structural
assumptions we impose for the GLRT are that $\Sigma_0(d)$ is diagonal (independence
under $\setH_0$) and that $\Sigma_1(d)$ has a free off-diagonal entry
(correlation under $\setH_1$).

In our model \eqref{eq:H0_model}--\eqref{eq:H1_model}, the relevant SOS parameters are% determined by $(\sigma_{\rnds}^2,\sigma_1^2,\sigma_2^2)$:
\begin{equation}
\Sigma_0(d) =
\begin{bmatrix}
\sigma_1^2 & 0 \\
0          & \sigma_2^2
\end{bmatrix},\quad
\Sigma_1(d) =
\begin{bmatrix}
\sigma_{\rnds}^2 + \sigma_1^2 & \sigma_{\rnds}^2 \\
\sigma_{\rnds}^2              & \sigma_{\rnds}^2 + \sigma_2^2
\end{bmatrix},
\end{equation}
so that the (population) correlation coefficient at delay $d$ is
\begin{equation}\label{eq:correlationcoefdef}
\rho_d =
\frac{\sigma_{\rnds}^2}{
\sqrt{(\sigma_{\rnds}^2+\sigma_1^2)(\sigma_{\rnds}^2+\sigma_2^2)}}.
\end{equation}
For the GLRT derivation, however, we only use the generic form $\Sigma_0(d),\Sigma_1(d)$ and treat their entries as unknown SOS parameters consistent with the above structure. Thus, let
\begin{align}
    \widehat{\sigma}^2_{\rndx} &\triangleq \frac{1}{N}\sum_{n=0}^{N-1} \rndx^2[n], \quad
    \widehat{\sigma}^2_{\rndy}(d) \triangleq \frac{1}{N}\sum_{n=0}^{N-1} \rndy^2[n+d],\\
    \widehat{\sigma}_{\rndx\rndy}(d) &\triangleq \frac{1}{N}\sum_{n=0}^{N-1} \rndx[n]\rndy[n+d],
\end{align}
denote the empirical variances and cross-covariance at delay $d$, respectively, and define the
empirical correlation coefficient
\begin{equation}
    \widehat{\rho}(d) \triangleq
    \frac{\widehat{\sigma}_{\rndx\rndy}(d)}{\widehat{\sigma}_{\rndx}\widehat{\sigma}_{\rndy}(d)}, \quad \widehat{\sigma}_{\rndx} \triangleq \sqrt{\widehat{\sigma}^2_{\rndx}}, \quad \widehat{\sigma}_{\rndy}(d) \triangleq \sqrt{\widehat{\sigma}^2_{\rndy}(d)}.
\end{equation}

For each $d \in \setD$, define the (maximized) likelihood ratio
\begin{equation}
    \Lambda_d(\setX_N,\setY_N^{\setD})
    \triangleq
    \frac{\displaystyle \sup_{\Sigma_1(d)} 
    p(\setX_N,\setY_N^{\setD} \mid \setH_1,d,\Sigma_1(d))}
    {\displaystyle \sup_{\Sigma_0(d)} 
    p(\setX_N,\setY_N^{\setD} \mid \setH_0,\Sigma_0(d))},
\end{equation}
where the numerator optimizes over all positive-definite $\Sigma_1(d)$ of the
form above, and the denominator optimizes over diagonal matrices
$\Sigma_0(d)$. The GLRT forms the statistic
\begin{equation}
    \Lambda_{\GLRT}(\setX_N,\setY_N^{\setD})
    \triangleq \max_{d \in \setD} \Lambda_d(\setX_N,\setY_N^{\setD})
\end{equation}
and compares it to a threshold:
\begin{equation}
    \Lambda_{\GLRT}(\setX_N,\setY_N^{\setD})
    \;\mathop{\gtrless}_{\setH_0}^{\setH_1}\; \eta.
\end{equation}

For Gaussian data, the maximum-likelihood estimates of the covariance matrices in the numerator and denominator coincide, respectively, with% the empirical covariance and its diagonal part, namely
\begin{equation}
    \widehat{\Sigma}_1(d) = 
    \begin{bmatrix}
        \widehat{\sigma}^2_{\rndx} & \widehat{\sigma}_{\rndx\rndy}(d) \\
        \widehat{\sigma}_{\rndx\rndy}(d) & \widehat{\sigma}^2_{\rndy}(d)
    \end{bmatrix},
    \quad
    \widehat{\Sigma}_0(d) =
    \begin{bmatrix}
        \widehat{\sigma}^2_{\rndx} & 0 \\
        0 & \widehat{\sigma}^2_{\rndy}(d)
    \end{bmatrix}.
\end{equation}
Substituting these estimates into the likelihoods under $\setH_0$ and $\setH_1$
and simplifying shows that, for each $d$, the ratio
$\Lambda_d(\setX_N,\setY_N^{\setD})$ is a strictly increasing function of
$\widehat{\rho}^2(d)$. Consequently, the GLRT is equivalent to a threshold test
\begin{equation}
    \max_{d \in \setD} \, \widehat{\rho}(d)
    \;\mathop{\gtrless}_{\setH_0}^{\setH_1}\; \gamma,
    \label{eq:GLRT_stat}
\end{equation}
for some threshold $\gamma$ chosen to meet a desired FA level. In simple words,
\emph{in the absence of communication constraints}, the GLRT reduces to a
variance-invariant correlator: it declares $\setH_1$ whenever the maximum
normalized empirical cross-correlation between the two sensors exceeds a
threshold.

To highlight the connection with the GLRT, consider the \emph{extremum-based empirical cross-correlation} estimator~\cite{weiss2025joint}
\begin{equation}\label{eq:rho_MIE_def}
    \widehat{\rho}_{\mie}(\ell)
    \;\triangleq\;
    \frac{\rndy[\rndj + \ell]}{\Exop[\rndx[\rndj]]}, \quad \forall \ell \in \setD,
\end{equation}
where the denominator is the mean of the maximum sample $\rndx[\rndj]$.\footnote{In the Gaussian setting, this mean can be expressed in closed-form in terms of the variance of $\rndx[n]$ and the blocklength $N$; see, e.g., related extremum-based correlation estimators in~\cite{weiss2025joint} and references therein. Its exact value is not needed for the structural
interpretation that follows.} Each $\widehat{\rho}_{\mie}(\ell)$ uses precisely one pair $(\rndx[\rndj],\rndy[\rndj+\ell])$ and can thus be viewed as a single-sample empirical correlation coefficient built from the (mean of the) encoder's maximum sample and the corresponding shifted sample at the decoder.

The behavior of $\widehat{\rho}_{\mie}(\ell)$ as an estimator of the underlying correlation is characterized in closely related settings in extremum-encoding-based time-delay estimation in~\cite{weiss2024joint,weiss2025extremum,weiss2025joint}. In particular, for standard jointly Gaussian processes with correlation coefficient $\rho$ and iid samples, one has
\begin{equation}\label{eq:rho_MIE_unbiased}
    \Exop\big[\widehat{\rho}_{\mie}(\ell)\big]
    = \rho,
\end{equation}
when $\Exop \left[\rndx[n]\rndy[n+\ell]\right] = \rho$ and the variance of $\widehat{\rho}_{\mie}(\ell)$ decays on the order of $1/k$ as the message length $k$ grows.

In our model, under $\setH_0$ the observations at the two sensors are
independent, so $\Exop[\rndx[n]\rndy[n+\ell]] = 0$ for all $\ell$, and hence
\begin{equation}
    \Exop\big[\widehat{\rho}_{\mie}(\ell)\mid \setH_0\big] = 0,
    \quad \forall\,\ell \in \setD.
\end{equation}
Under $\setH_1$ with true time delay $d$, the cross-correlation between the aligned
pairs satisfies
\begin{equation}
    \Exop\left[\rndx[n]\rndy[n+\ell]\mid \setH_1\right] = \sigma_{\rnds}^2 \cdot\mathbbm{1}_{\{\ell=d\}},
\end{equation}
so that the corresponding correlation coefficient is $\rho_d$ \eqref{eq:correlationcoefdef} for $\ell=d$ and $0$ for $\ell\neq d$. Hence,
\begin{equation}
    \Exop\big[\widehat{\rho}_{\mie}(\ell)\mid \setH_1\big]
    =
    \begin{cases}
        \rho_d, & \ell = d,\\[0.05cm]
        0,      & \ell \neq d.
    \end{cases}
\end{equation}

Since $\Exop[\rndx[\rndj]]$ does not depend on $\ell$, the test based on
$\rndt_{\midsub}$ can be equivalently expressed as a threshold test on the maximum
of these correlation estimates. Indeed, from~\eqref{eq:T_MIE_repeat} and
\eqref{eq:rho_MIE_def},
\begin{equation}
    \max_{\ell \in \setD} \rndy[\rndj + \ell]
    = \Exop[\rndx[\rndj]] \cdot \max_{\ell \in \setD} \widehat{\rho}_{\mie}(\ell),
\end{equation}
so that
\begin{equation}
    \rndt_{\midsub} \;\mathop{\gtrless}_{\setH_0}^{\setH_1}\; \tau
    \quad\Longleftrightarrow\quad
    \max_{\ell \in \setD} \widehat{\rho}_{\mie}(\ell)
    \;\mathop{\gtrless}_{\setH_0}^{\setH_1}\; \gamma,
\end{equation}
with $\gamma \triangleq \tau / \Exop[\rndx[\rndj]]$. Thus, \emph{our detector \eqref{eq:miedecisionrule} is also a correlation test}: it compares the maximum
over $\ell$ of a suitably scaled empirical cross-correlation coefficient to a
threshold.

Comparing this with the GLRT statistic~\eqref{eq:GLRT_stat}, we see that, in the absence of communication constraints, the GLRT uses all $N$ aligned samples to form normalized empirical cross-correlations $\widehat{\rho}(d)$ and thresholds $\max_{d\in\setD}\widehat{\rho}(d)$, whereas our scheme compresses the observation of the encoder into the single extremum index $\rndj$ and then forms one-sample (extremum-based) empirical correlations $\widehat{\rho}_{\mie}(\ell)$ at the decoder, thresholding $\max_{\ell\in\setD}\widehat{\rho}_{\mie}(\ell)$. The two rules thus share the same core structure, differing only in how the empirical cross-correlation is estimated under the $k$-bit communication constraint.

% \vspace{-0.4cm}
\subsection{Nonasymptotic Performance Analysis}\label{subsec:nonasymptoticanalysis}

We now characterize the performance of the extremum-based detector. Specifically, for a fixed threshold~$\tau$, we derive nonasymptotic closed-form expressions for the resulting FA and MD probabilities under the statistical model \eqref{eq:H0_model}--\eqref{eq:H1_model} presented in Section~\ref{sec:problemformulation}. Throughout this subsection, we consider the decision rule is given by~\eqref{eq:miedecisionrule}.

\subsubsection*{False-alarm probability}

Under $\setH_0$, the processes at the two sensors are independent, and the samples $\{\rndy[\rndj+\ell]\}_{\ell\in\setD}$ used by the detector are iid Gaussian with variance~$\sigma_2^2$. The FA event is therefore the event that at least one of these $2d_m+1$ samples exceeds the threshold~$\tau$. This yields the following result.

\begin{proposition}[FA probability]
For the decision rule~\eqref{eq:miedecisionrule}, the FA probability $P_{\fa}(\tau) \triangleq \prob{ \delta(\rvecm, \setY_N) = \setH_1 \mid \setH_0}$ is given by
\begin{align}
        P_{\fa}(\tau) &= 1 - \left(1-Q\left(\frac{\tau}{\sigma_2}\right)\right)^{2d_m+1},\label{eq:fa_closed}
\end{align}
where $Q(x)\triangleq \frac{1}{\sqrt{2\pi}}\int_{x}^{\infty}e^{-\frac{t^2}{2}}{\rm d}t$
is the standard $Q$-function.
\end{proposition}
\begin{IEEEproof}
We have
\begin{align}
        % P_{\fa} &\triangleq \prob{ \delta(\rvecm, \setY_N) = \setH_1 \mid \setH_0}\\
        P_{\fa}(\tau) &= \prob{\rndt_{\midsub} \geq \tau \mid \setH_0}\\
        &= \prob{\max_{\ell \in \setD} \; \rndy[\rndj + \ell] \geq \tau \mid \setH_0}\\
        % &= \prob{\max_{\ell \in \setD} \; \rndz_2[\rndj + \ell] \geq \tau}\\
        &= 1 - \prob{\max_{\ell \in \setD} \; \rndz_2[\rndj + \ell] \leq \tau}\\
        &= 1 - \prob{\rndz_2[n] < \tau}^{|\setD|}\\
        &= 1 - \left(1-Q\left(\frac{\tau}{\sigma_2}\right)\right)^{2d_m+1},
\end{align}
where we have used the fact that, under $\setH_0$, $\rndy[n]=\rndz_2[n]$.
\end{IEEEproof}

\subsubsection*{Misdetection probability}

Under $\setH_1$, the analysis is more involved due to the statistical dependence
between the transmitted extremum index~$\rndj$ and the samples of $\rndy[\cdot]$
in its neighborhood, as well as due to the composite nature of the alternative
hypothesis (unknown delay~$d$). Still, we may derive an exact expression for the MD probability, as follows.

\begin{proposition}[MD probability]
Let
    \begin{align}
    \sigma_{\rndx}^2 &\triangleq \sigma_{\rnds}^2+\sigma_1^2, \quad\;\;\;\, \beta \triangleq \frac{\sigma_{\rnds}^2}{\sigma_{\rndx}^2},\\
    \sigma_{\mmse}^2 &\triangleq \sigma_2^2 + \beta\sigma_1^2, \;\;\; \sigma_{\eff}^2 \triangleq \beta^2\sigma_{\rndx}^2+\sigma_{\mmse}^2,\\
    M_{\mathrm{in}}(j) &\triangleq \min(d_m,j) + \min(d_m,N-1-j),\label{eq:numberofinside}\\
    M_{\mathrm{out}}(j) &\triangleq 2d_m - M_{\mathrm{in}}(j),\label{eq:numberofoutside}
    \end{align}
and define $\Phi(\cdot)$ as the standard Gaussian cumulative distribution function (CDF),
$\Phi_2(\cdot,\cdot;\rho)$ as the CDF of a zero-mean bivariate Gaussian with
unit variances and correlation coefficient~$\rho$, and
$\Phi^{-1}(\cdot)$ as the inverse of~$\Phi(\cdot)$. Then, for the decision
rule~\eqref{eq:miedecisionrule}, the MD probability $P_{\md}(\tau) \triangleq \prob{ \delta(\rvecm, \setY_N) = \setH_0 \mid \setH_1}$ is given by
\begin{align}
    P_{\md}(\tau) &=\int_{0}^{1} N u^{N-1}
    \left\{\frac{1}{N}\sum_{j=0}^{N-1}
    Q\left(\tfrac{\beta\sigma_{\rndx}\Phi^{-1}(u)-\tau}{\sigma_{\mmse}}\right) \right. \nonumber\\
    &\times\left.\left[\tfrac{\Phi_2\left(\Phi^{-1}(u),\,\tfrac{\tau}{\sigma_{\eff}};\, \rho\right)}{u}\right]^{M_{\mathrm{in}}(j)}\left[\Phi\left(\tfrac{\tau}{\sigma_{\eff}}\right)\right]^{M_{\mathrm{out}}(j)}
    \right\} {\rm d}u.\label{eq:p_md_exact}
\end{align}
\end{proposition}

\begin{IEEEproof}
Using the linear minimum mean-square error decomposition of $\rnds[n]$ based on $\rndx[n]$, we may write
\begin{equation}
\rndy[n] = \beta\rndx[n-d] + \mybar{\rndz}[n],\label{eq:mmse-dec}
\end{equation}
with $\mybar{\rndz}[n]\overset{\text{iid}}{\sim} \normal(0,\sigma_{\mmse}^2)$ independent of $\{\rndx[n]\}_{n\in\integers}$. By symmetry, $\rndj$ is uniform on $\{0,\ldots,N-1\}$ and independent of $\rndx[\rndj]$. Moreover, given $\{\rndj=j,\rndx[\rndj]=x\}$, the $N-1$ remaining samples are iid as $\rndv\mid \rndv\le x$ with $\rndv\sim\normal(0,\sigma_{\rndx}^2)$. Therefore,
\begin{align}
        &P_{\md}=\Exop\left[ \prob{\max_{\ell \in \setD} \; \rndy[\rndj + \ell] < \tau \mid \setX_N, \setH_1} \mid \setH_1 \right]\\
        &=\Exop\left[ \prod_{\ell\in\setD}\prob{\beta\rndx[\rndj + \ell - d] - \tau < - \mybar{\rndz}[\rndj + \ell] \mid \setX_N, \setH_1} \mid \setH_1 \right]\\
        &=\Exop\left[\underbrace{Q\left( \tfrac{ \beta\rndx[\rndj] - \tau }{ \sigma_{\mmse} } \right)}_{\text{``aligned"}}
\cdot
\prod_{\substack{\ell\in\setD \\ \ell\neq d}}
\underbrace{Q\left( \tfrac{ \beta\rndx[\rndj+\ell-d] - \tau }{ \sigma_{\mmse} } \right)}_{\text{``non-aligned"}} \mid \setH_1
\right].\label{eq:expectationofalignedandnonaligned}
\end{align}

Now, consider the ``inside/outside" structure for the non-aligned lags. Fix $j\in\{0,\ldots,N-1\}$ and write $r\triangleq \ell-d\in\setD\backslash\{0\}$. Then $\rndj+r$ either falls \emph{inside} the encoded block $\{0,\ldots,N-1\}$, or \emph{outside}:

\begin{itemize}
    \item \emph{Inside lags ($\rndj+r\in\{0,\ldots,N-1\}$, $r\neq 0$)}: Conditioned on $\{\rndj=j, \rndx[\rndj]=x\}$, the $M_{\mathrm{in}}(j)$ values $\{\rndx[\rndj+r]\}_{\text{inside}}$ are iid as $\rndv\mid \rndv\le x$ with $\rndv\sim\normal(0,\sigma_{\rndx}^2)$. Using \eqref{eq:mmse-dec}, for a generic inside lag,
\begin{align}
&\Exop\left[ Q\left(\frac{\beta\rndx[\rndj+r]-\tau}{\sigma_{\mmse}}\right) \mid \rndx[\rndj]=x, \rndj=j, \setH_1 \right]\\
&= \Exop\left[ Q\left(\frac{\beta\rndv-\tau}{\sigma_{\mmse}}\right) \mid \rndv\le x, \setH_1 \right]
\triangleq B_{\mathrm{in}}(x,\tau).
\end{align}
Since $\{\rndx[\rndj+r]\}_{\text{inside}}$ are iid given $x$, the $M_{\mathrm{in}}(j)$ (defined in \eqref{eq:numberofinside}) inside factors multiply to $B_{\mathrm{in}}(x,\tau)^{M_{\mathrm{in}}(j)}$.

    \item \emph{Outside lags ($\rndj+r\notin\{0,\ldots,N-1\}$)}: These $\rndx[\rndj+r]$ are independent of $\mathcal{X}_N$---and therefore of $\rndj$ and $\rndx[\rndj]$---with law $\normal(0,\sigma_{\rndx}^2)$. Hence, for a generic outside lag,
\begin{align}
&\Exop\left[ Q\left(\frac{\beta\rndx[\rndj+r]-\tau}{\sigma_{\mmse}}\right) \mid \rndx[\rndj]=x, \rndj=j, \setH_1 \right]\\
&= \Exop\left[ Q\left(\frac{\beta\rndv-\tau}{\sigma_{\mmse}}\right) \right] \triangleq B_{\mathrm{out}}(\tau),
\end{align}
so the $M_{\mathrm{out}}(j)$ \eqref{eq:numberofoutside} factors multiply to $B_{\mathrm{out}}(\tau)^{M_{\mathrm{out}}(j)}$.

\end{itemize}

% We now turn to compute the closed-form expressions for inside and outside building blocks, $B_{\mathrm{in}}(x,\tau)$ and $B_{\mathrm{out}}(\tau)$, respectively.
We now turn to compute $B_{\mathrm{in}}(x,\tau)$ and $B_{\mathrm{out}}(\tau)$. Let $\rndv\sim\normal(0,\sigma_{\rndx}^2)$ and $\rndg\sim\normal(0,1)$ independent, and define the utility random variable $\rndu\triangleq \beta\,\rndv-\sigma_{\mmse}\cdot\rndg\sim\normal(0,\sigma_{\eff}^2)$. Then,
\begin{align}
B_{\mathrm{out}}(\tau) &= \Exop\left[ \prob{\rndu<\tau \mid \rndv} \right] = \prob{\rndu<\tau} = \Phi\left(\frac{\tau}{\sigma_{\eff}}\right). \label{eq:Bout}
\end{align}
For the inside factor, by Bayes’ rule and standardization,
\begin{align}
% B_{\mathrm{in}}(x,\tau) &= \prob{\rndu<\tau \mid \rndv\le x} \\
% &= \frac{\prob{\rndu<\tau, \rndv\le x}}{\prob{\rndv\le x}} = \frac{\Phi_2\left(\frac{x}{\sigma_{\rndx}}, \frac{\tau}{\sigma_{\eff}}; \rho\right)}{\Phi\left(\frac{x}{\sigma_{\rndx}}\right)}, \label{eq:Bin}
B_{\mathrm{in}}(x,\tau) &= \frac{\prob{\rndu<\tau, \rndv\le x}}{\prob{\rndv\le x}} = \frac{\Phi_2\left(\frac{x}{\sigma_{\rndx}}, \frac{\tau}{\sigma_{\eff}}; \rho\right)}{\Phi\left(\frac{x}{\sigma_{\rndx}}\right)}, \label{eq:Bin}
\end{align}
where $\rho \triangleq \frac{\beta\sigma_{\rndx}}{\sigma_{\eff}}$ is the correlation coefficient between $\rndv$ and $\rndu$, and $\Phi_2(t_1,t_2;\rho) \triangleq \int_{-\infty}^{t_1} \phi(v)\Phi\left(\tfrac{t_2-\rho v}{\sqrt{1-\rho^2}}\right){\rm d}v$ is the standard bivariate normal CDF with correlation $\rho$.

Having established the conditional expectations given $\{\rndj=j, \rndx[\rndj]=x\}$, we are left with the task of averaging over the maximum and the extremum location. Using the standard notation $\phi(x)\triangleq \frac{1}{\sqrt{2\pi}}e^{-\frac{x^2}{2}}$, the probability density function (PDF) of the block maximum $\rndx[\rndj]$ is given by
\begin{equation}
p_{\rndx[\rndj]}(x) = \tfrac{N}{\sigma_{\rndx}}\phi\left(\tfrac{x}{\sigma_{\rndx}}\right)\Phi\left(\tfrac{x}{\sigma_{\rndx}}\right)^{N-1}, \quad x\in\reals,
\label{eq:pdf-max}
\end{equation}
and $\prob{\rndj=j}=\tfrac{1}{N}$ for all $j$. Plugging the ``aligned" and ``non-aligned" factors into \eqref{eq:expectationofalignedandnonaligned} and averaging yields
\begin{align}
\hspace{-0.2cm}P_{\md}(\tau)
= \frac{1}{N}\sum_{j=0}^{N-1} \int_{-\infty}^{\infty}
&\underbrace{Q\!\left(\frac{\beta\,x-\tau}{\sigma_{\mmse}}\right)}_{\text{aligned}}\cdot
\underbrace{\big[B_{\mathrm{in}}(x,\tau)\big]^{M_{\mathrm{in}}(j)}}_{\text{inside lags}} \nonumber \\
&\times
\underbrace{\big[B_{\mathrm{out}}(\tau)\big]^{M_{\mathrm{out}}(j)}}_{\text{outside lags}}
\, p_{\rndx[\rndj]}(x)\, {\rm d}x.
\label{eq:Pmd-exact}
\end{align}
Using the change of variables $u=\Phi(x/\sigma_{\rndx})$, hence $x=\sigma_{\rndx}\,\Phi^{-1}(u)$ and $p_{\rndx[\rndj]}(x)\,{\rm d}x=N\,u^{N-1}{\rm d}u$, gives \eqref{eq:p_md_exact}.
\end{IEEEproof}

\vspace{-0.25cm}
\section{Simulation Results}
\begin{figure}[t]
        \centering
        \includegraphics[width=0.93\columnwidth]{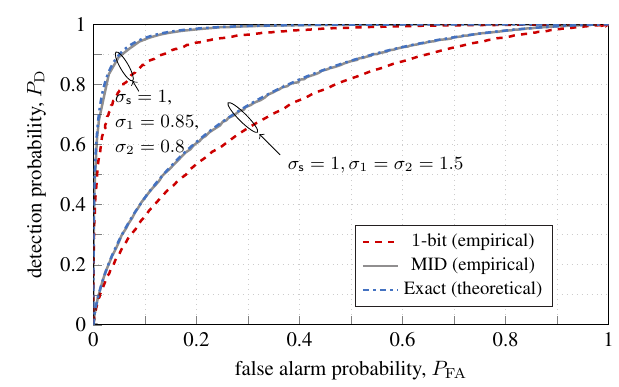}
        \caption{ROC curve for $k=8$ and $d_m=50$.}
        \label{fig:roc_k8}\vspace{-0.4cm}
    \end{figure}
\begin{figure}[t]
        \centering
        \includegraphics[width=0.93\columnwidth]{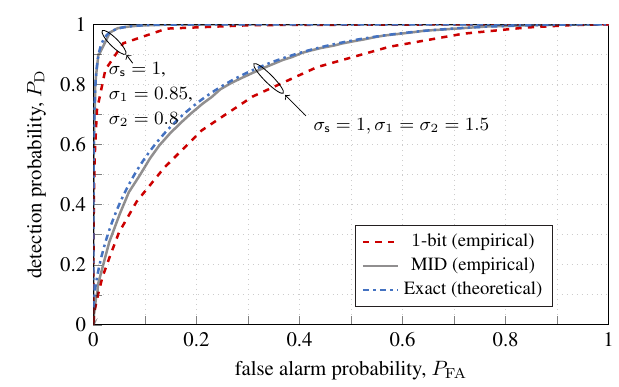}
        \caption{ROC curve for $k=12$ and $d_m=150$.}
        \label{fig:roc_k12}\vspace{-0.5cm}
    \end{figure}
\begin{figure}[t]
        \centering
        \includegraphics[width=0.93\columnwidth]{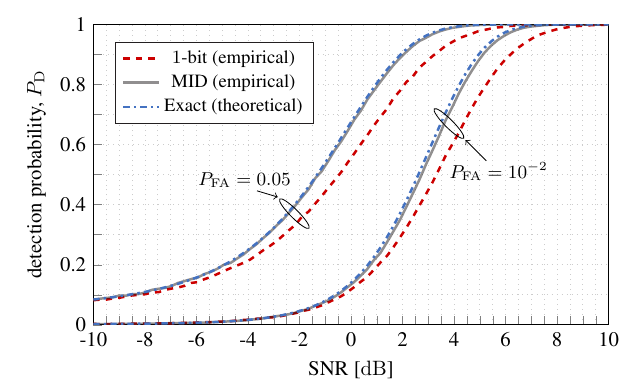}
        \caption{Probability of detection $P_{\mathrm{D}}$ as a function of the SNR in $\mathrm{dB}$ for fixed $P_{\mathrm{FA}}$, $\sigma_1=\sigma_2=1$, $k=8$ and $d_m=50$.}
        \label{fig:snr_pd_fixed_pfa}\vspace{-0.4cm}
    \end{figure}
We compare the empirical performance of the MID \eqref{eq:miedecisionrule} with its analytically predicted $P_{\fa}$ in~\eqref{eq:fa_closed} and $P_{\md}$ in~\eqref{eq:p_md_exact}, as well as with a $1$-bit-per–sample-based (e.g.,~\cite{weiss2021onebit}) quantization detector. In all simulations, we choose the admissible maximum delay $d_m=\{50,150\}$ for $k=\{8,12\}$, respectively, fix the true delay to $d=32$ and set $N=2^k$. The receiver operating characteristic (ROC) curves for $k=\{8,12\}$ (Figs.~\ref{fig:roc_k8} and~\ref{fig:roc_k12}) show an effectively perfect match between the theoretical and empirical MID performance across all FA levels. Additionally, the MID consistently outperforms the $1$-bit-quantization-based detector, especially in the moderate-FA regime. Finally, SNR sweeps at fixed $P_{\fa}\in\{0.05,10^{-2}\}$ (Fig.~\ref{fig:snr_pd_fixed_pfa}), where  $\sigma_1=\sigma_2=1$ so that SNR$=\sigma_{\rnds}^2$, confirm that the empirical results track the theoretical prediction closely, and the $1$-bit scheme continues to incur a consistent SNR penalty, highlighting the advantages of our proposed scheme.

\vspace{-0.2cm}
\section{Concluding Remarks}\label{sec:conclusion}
We studied joint compression and detection in a distributed sensing setting with a limited $k$-bit link from a distant sensor to another, co-located with the computing unit. We proposed a maximum-index-based detector, showed that it can be interpreted as a correlation test closely related to the centralized GLRT, and derived exact nonasymptotic expressions for its FA and MD probabilities. Simulation results corroborate these analytical performance curves, and further indicate that the extremum-based scheme provides competitive detection performance despite under tight communication constraint. These findings highlight extremum encoding as a promising primitive for distributed inference under stringent rate limitations.

% \begin{figure}
%     \centering
%     \includesvg[width=0.78\columnwidth]{figures/ROC_figure_for_EA_method_SNR_10_dB.svg}
%     \caption{ROC curves of the proposed method and an alternative competitor method, which applies one-bit-per-sample ($1$-bit) quantization (see, e.g.,~\cite{weiss2024joint}), for an SNR level of $10$dB, where $P_{\detect}\triangleq1-P_{\md}$. Evidently, our method considerably outperforms the common, though na\"ive, method of $1$-bit compression.}
%     \label{fig:simulationresults}
% \end{figure}

\vspace{-0.2cm}

\bibliographystyle{IEEEbib}
\bibliography{./Inputs/refs}

\end{document}